%% file: AtomicAppends-arXiv.tex
\documentclass[a4paper,english]{article}
\usepackage{fullpage}
\usepackage{amssymb,amsmath,multicol}
\usepackage{graphicx,url}
\usepackage{color,setspace,enumitem}
\usepackage{epstopdf,xspace}
\usepackage{multirow}

\usepackage[ruled]{algorithm}
\usepackage[noend]{algpseudocode}
\usepackage[normalem]{ulem}
\usepackage{pifont}

\usepackage{wrapfig}
\usepackage{cite}

\newtheorem{theorem}{Theorem}%[section]
\newtheorem{corollary}{Corollary}%[theorem]

\newtheorem{definition}{Definition}
\newenvironment{proof}{\noindent{\bf Proof.}}{\hfill$\Box$} 

\usepackage{etoolbox}\AtBeginEnvironment{algorithmic}{\small}
\floatname{algorithm}{{\small Code}}

\newcommand{\xmark}{\ding{56}}%

\algblockdefx[Receive]{Receive}{EndReceive}%
[1]{{\bf receive}   (#1) from process $\pr$}%
{{\bf end receive}}

\algblockdefx[Upon]{Upon}{EndUpon}%
[1]{{\bf upon} (#1) do}%
{{\bf end upon}}

\algdef{SE}[SUBALG]{Indent}{EndIndent}{}{\algorithmicend\ }%
\algtext*{Indent}
\algtext*{EndIndent}

\makeatletter
\ifthenelse{\equal{\ALG@noend}{t}}%
{\algtext*{EndReceive}}
{}%

\makeatletter
\ifthenelse{\equal{\ALG@noend}{t}}%
{\algtext*{EndUpon}}
{}%

\algrenewcommand{\ALG@beginalgorithmic}{\small}
\algrenewcommand\alglinenumber[1]{\small #1:}

\makeatother

\input{nn_macros}
\newcommand{\dledger}{\mathcal{DL}}
\newcommand{\mdledger}{\mathcal{MDL}}
\newcommand{\DLOA}{DLO$_{A}$\xspace}
\newcommand{\DLOB}{DLO$_{B}$\xspace}
\newcommand{\DLOX}{DLO$_{X}$\xspace}
\newcommand{\DLOBX}{DLO$_{\bar{X}}$\xspace}

\newcommand{\extends}{\Vert}

\definecolor{mygreen}{RGB}{0, 125, 125}
\newcommand{\nn}[1]{\textcolor{black}{#1}}
\newcommand{\cg}[1]{\textcolor{black}{#1}}
\newcommand{\af}[1]{\textcolor{black}{#1}}

\begin{document}

\title{Atomic Appends: Selling Cars and Coordinating Armies with Multiple Distributed Ledgers\thanks{Partially supported by the Spanish Ministry of Science, Innovation and Universities grant DiscoEdge (TIN2017-88749-R), the Regional Government of Madrid (CM) grant Cloud4BigData
	(S2013/ICE-2894) co-funded by FSE \& FEDER, the NSF of China grant
	61520106005, and by funds for the promotion of research at the University of Cyprus.}}

%\titlerunning{}

\author{
	Antonio Fern\'andez Anta\thanks{IMDEA Networks Institute, 
	Madrid, Spain. Email: antonio.fernandez@imdea.org}
\and 
	Chryssis Georgiou\thanks{Dept. of Computer Science, University of Cyprus,
	Nicosia, Cyprus. Email: chryssis@cs.ucy.ac.cy}
\and
    Nicolas Nicolaou\thanks{KIOS Research and Innovation CoE, University of Cyprus \& Algolysis Ltd, Cyprus. email: nicolasn@ucy.ac.cy}
}

%\authorrunning{A. Fern\'andez Anta, Ch. Georgiou, and N. Nicolaou}%mandatory. First: Use abbreviated first/middle names. Second (only in severe cases): Use first author plus 'et al.'

%\Copyright{Antonio Fern\'andez Anta, Chryssis Georgiou and Nicolas Nicolaou}

%\subjclass{\ccsdesc[500]{Theory of computation~Design and analysis of algorithms~Distributed algorithms}}

%\keywords{Distributed Ledgers, Interoperability, Atomic Appends, Rational Clients, Fault-tolerance}
 
%\category{\af{Regular paper}}

%\funding{Partially supported by the Spanish Ministry of Science, Innovation and Universities grant DiscoEdge (TIN2017-88749-R), the Regional Government of Madrid (CM) grant Cloud4BigData
%	(S2013/ICE-2894) co-funded by FSE \& FEDER, the NSF of China grant
%	61520106005, and by funds for the promotion of research at the University of Cyprus.}

%\acknowledgements{We would like to thank Kishori Konwar, Michel Raynal, and Gregory Chockler for
%	insightful discussions.}

%\EventEditors{xxx}
%\EventNoEds{0}
%\EventLongTitle{xxx}
%\EventShortTitle{xxx}
%\EventAcronym{xxx}
%\EventYear{2018}
%\EventDate{xxx}
%\EventLocation{xxx}
%\EventLogo{}
%\SeriesVolume{x}
%\ArticleNo{x}

\maketitle

\begin{abstract}
The various applications using Distributed Ledger Technologies (DLT) or blockchains, \af{have led}
%lead 
to the introduction of a new ``marketplace'' where multiple types of digital assets may be exchanged. 
As each blockchain is designed to support 
\af{specific types of assets and transactions, and no blockchain will prevail,}
%a specific type of transaction, 
the need to perform \textit{interblockchain} transactions is 
\af{already pressing.}
%not very far ahead. 

\sloppy{In this work we examine the fundamental problem of interoperable and interconnected 
\af{blockchains.}}
%DLT. 
In particular, we begin by \af{introducing}
%defining 
the \textit{Multi-Distributed Ledger Objects} (MDLO), which 
is the result of aggregating multiple \textit{Distributed Ledger Objects} -- DLO (a DLO is a formalization of the blockchain) and that 
supports append and get operations of records (e.g., transactions) in them from multiple 
clients concurrently. Next we define the \textit{AtomicAppends} problem, which emerges
when the exchange of digital assets between multiple clients may involve \af{appending records in} more than one \af{DLO.}
%DLT. 
Specifically, AtomicAppend 
requires that either \emph{all} 
\af{records}
%digital assets 
will be appended on the involved 
\af{DLOs}
%DLTs 
or \emph{none}. We examine the solvability of this 
problem \cg{assuming {\em rational and risk-averse} clients that may {\em fail by crashing}, and under different %various 
	client {\it utility} and {\it append} models, {\it timing models}, and  
	%(i)  client \emph{collaboration} models, (ii) \emph{synchrony} models, and (iii) 
	client \emph{failure scenarios}.}  
We show that for some cases the existence of a intermediary is {\em necessary} for the problem solution. We propose
the implementation of such intermediary over a specialized blockchain, we term \textit{Smart DLO} (SDLO), and we show 
how this can be used to solve the AtomicAppends problem \cg{even} in an asynchronous, client competitive  environment, where
all the clients may crash. 
\end{abstract}

\section{Introduction}
\label{sec:Intro}
\input{introduction.tex}

\section{Problem Statements and Model of Computation}
\label{model}
\input{model.tex}

\section{Crash-prone Collaborative AtomicAppends with Client Appends}
\label{nosdlo}
\input{collaborative_noSDLO.tex}

\section{Crash-prone AtomicAppends with SDLO}
\label{sdlo}
\input{collaborative_SDLO.tex}	

%\section{Competitive AtomicAppends (with SDLO)}
%\label{competitive}

\section{Conclusion}
\label{conclusion}
\input{conclusion.tex}

\paragraph{Acknowledgements.} We would like to thank Kishori Konwar, Michel Raynal, and Gregory Chockler for
insightful discussions.

%%%% REFERENCES %%%%
%\bibliographystyle{plain}
%\bibliographystyle{plainurl}
%\bibliography{biblio}

\end{document}

%% file: nn_macros.tex
\providecommand{\tup}[1]{%
    \relax\ifmmode
      \langle #1 \rangle%
    \else
        $\langle$#1$\rangle$%
    \fi
}

\newcommand{\act}[1]{%
    \relax\ifmmode
        \mathord{\mathcode`\-="702D\sf #1\mathcode`\-="2200}%
    \else
        $\mathord{\mathcode`\-="702D\sf #1\mathcode`\-="2200}$%
    \fi
}

\newcommand{\remove}[1]{}

\makeatletter
\def\mainlistofsymbols{
  \normalsize
  \vspace*{1.5 em}
  \@starttoc{los}
}

\def\partonelistofsymbols{
  \normalsize
  \vspace*{1.5 em}
  \@starttoc{p1los}
}

\def\parttwolistofsymbols{
  \normalsize
  \vspace*{1.5 em}
  \@starttoc{p2los}
}

\def\l@symbol#1#2{\addpenalty{-\@highpenalty} \vskip 4pt plus 2pt
{\@dottedtocline{0}{0em}{8em}{#1}{#2}}}
\makeatother

\newcommand{\newhiddensym}[2]{%
}

\newcommand{\algIOA}[2]{\ifmmode{\text{#1}_{#2}}\else{$\text{#1}_{#2}$}\fi}
\newcommand{\EX}{\ifmmode{\xi}\else{$\xi$}\fi}
\newcommand{\EXF}{\ifmmode{\phi}\else{$\phi$}\fi}
\newcommand{\hist}[1]{H_{#1}}
\newcommand{\obj}[1]{O_{#1}}

\newcommand{\inter}[1]{
	\ifmmode{\left(\bigcap_{\mathcal{Q}\in#1}\mathcal{Q}\right)}
	\else{$\left(\bigcap_{\mathcal{Q}\in#1}\mathcal{Q}\right)$}
	\fi
}
\newcommand{\srvSet}{\mathcal{S}}
\newcommand{\ledger}{\mathcal{L}}
\newcommand{\op}{\pi}
\mathchardef\mhyphen="2D
\newcommand{\pr}{p}
\newcommand{\rdr}{r}
\newcommand{\bef}{\rightarrow}

\newcommand{\vid}[1]{\ifmmode{\nu_{#1}}\else{$\nu_{#1}$}\fi}
\newcommand{\seen}{\ifmmode{seen}\else{$seen$}\fi}
\newcommand{\maxts}[1]{\ifmmode{maxTS_{#1}}\else{$maxTS_{#1}$}\fi}
\newcommand{\maxtag}[1]{\ifmmode{maxTag_{#1}}\else{$maxTag_{#1}$}\fi}
\newcommand{\maxpair}[1]{\ifmmode{maxMPair_{#1}}\else{$maxMPair_{#1}$}\fi}
\newcommand{\mintag}[1]{\ifmmode{minTag_{#1}}\else{$minTag_{#1}$}\fi}
\newcommand{\maxps}{\ifmmode{maxPS}\else{$maxPS$}\fi}
\newcommand{\conftg}[1]{\ifmmode{confirmed_{#1}}\else{$confirmed_{#1}$}\fi}
\newcommand{\maxconftag}{\ifmmode{\ms{maxCT}}\else{$maxCT$}\fi}

%% file: introduction.tex
%!TEX root =  SCNDS18.tex

\subsection{Motivation}

Blockchain systems, cryptocurrencies, and distributed ledger technology (DLT) in general, are becoming very popular and are expected to have a high impact in multiple aspects of our everyday life. In fact, there is a growing number of applications that use DLT to support their operations \cite{Zago2018}. However, there are many different blockchain systems, and new ones are proposed almost everyday. Hence, it is extremely unlikely that one single DLT or blockchain system will prevail. This is forcing the DLT community to accept that it is inevitable to come up with ways to make blockchains interconnect and interoperate.

\cg{The work in~\cite{FGKN_NETYS18} proposed a} formal definition of a reliable concurrent object, termed Distributed Ledger Object (DLO), which tries to convey the essential elements of blockchains. In particular, a DLO is a sequence of records, and has only two operations, \act{append} and \act{get}. The \act{append} operation is used to attach a new record at the end of the sequence, while the \act{get} operation returns the sequence.

In this work we initiate the study of systems formed by multiple DLOs that interact among each other. To do so, we define a basic problem involving two DLOs, that we call {\em the Atomic Append problem}. In this problem, two clients want to append new records in two DLOs, so that either both records are appended or none. The clients are assumed to be selfish, \cg{but  rational and risk-averse~\cite{osborne2004introduction}}, and may have different incentives for the different outcomes. 
Additionally, we assume that they may fail by crashing, which makes solving the problem more challenging. We observe that the problem cannot be solved in some system models and propose algorithms that solve it in others.

\subsection{Related Work}

The Atomic Append problem we describe above is very related to the multi-party fair exchange problem \cite{DBLP:conf/fc/FranklinT98}, in which several parties exchange commodities so that everyone gives an item away and receives an item in return. The proposed solutions for this problem rely on cryptographic techniques \cite{DBLP:conf/focs/MicaliRK03,DBLP:conf/fc/MukhamedovKR05} and are not designed for distributed ledgers. In this paper, as much as possible, we want to solve Atomic Appends on DLOs via their two operations \act{append} and \act{get}, without having to rely on cryptography or smart contracts.

Among the first problems identified involving the interconnection of blockchains was Atomic Cross-chain Swaps~\cite{DBLP:conf/podc/Herlihy18}, which can also be seen as a version of the fair exchange problem. In this case, two or more users want to exchange assets (usually cryptocurrency) in multiple blockchains. This problem can be solved by using escrows, hashlocks and timelocks: all assets are put in escrow until a value $x$ with a special hash $y=hash(x)$ is revealed or a certain time has passed. Only one of the users knows $x$, but as soon as she reveals it to claim her assets, everyone can use it to claim theirs. Observe that this solution assumes synchrony in the system.

This technique was originally proposed in on-line fora for two users \cite{Atomic-swap},
%\footnote{\url{https://en.bitcoin.it/wiki/Atomic_cross-chain_trading}, accessed Nov 14, 2018.}, 
and it has been extensively adapted and used \cite{DBLP:books/daglib/0040621}. For instance, the Interledger system \cite{Interledger}
%\footnote{\url{https://interledger.org/}, accessed Nov 14, 2018.} 
will use a generalization of atomic swaps to transfer (and exchange) currency in a network of blockchains and connectors, allowing any client of the system to interact with any other client. The Lighting network \cite{poon2016bitcoin,miller2017sprites} also allows transfers between any two clients via a network of micro-payment channels using a generalized atomic swap. Both Interledger and Lighting route and create one-to-one transfer paths in their respective networks. Herlihy \cite{DBLP:conf/podc/Herlihy18} has formalized and generalized atomic cross-chain swaps beyond one-to-one paths, and shows how multiple cross-chain swaps can be achieved if the transfers form a strongly connected directed graph.

Unlike in most blockchain systems, in Hyperledger Fabric \cite{DBLP:conf/eurosys/AndroulakiBBCCC18,DBLP:conf/esorics/AndroulakiCCK18} it is possible to have transactions that span several blockchains (blockchains are called \emph{channels} in Hyperledger Fabric). This allows solving the atomic cross-chain swap problem using a third trusted channel or a mechanism similar to a two-phase commit \cite{DBLP:conf/esorics/AndroulakiCCK18}. Additionally, these solutions do not require synchrony from the system. The ability of channels to access each other's state and interact is a very interesting feature of Hyperledger Fabric, very in line with the techniques we assume from advanced distributed ledgers in this paper. Unfortunately, they seem to be limited to the channels of a given Hyperledger Fabric deployment.

There are other blockchain systems under development that, like Hyperledger Fabric, will allow interactions between the different chains, presumably with many more operations than atomic swaps. Examples are Cosmos \cite{Cosmos}
%\footnote{\url{https://cosmos.network}, accessed Nov 22, 2018.} 
or PolkaDot \cite{PolkaDot}.
%\footnote{\url{https://polkadot.network}, accessed Nov 22, 2018.}. 
These systems will have their own multi-chain technology, so only chains in a given deployment can initially interact, and other blockchain will be connected via gateways.

Another proposal for interconnection of blockchains is Tradecoin \cite{DBLP:journals/corr/abs-1805-05934}, whose target is to interconnect all blockchains by means of gateways, trying to reproduce the way Internet works. since the gateways will be clients of the blockchains, the functionality of the global interledger system will be limited by what can be done from the edge of the blockchains (i.e., by the blockchains' clients).

The practical need of blockchain systems to access the outside world to retrieve data (e.g., exchange rates, bank account balances) has been solved with the use of \emph{blockchain oracles}. These are relatively reliable sources of data that can be used inside a blockchain, typically in a smart contract. The weakest aspect of blockchain oracles is trust, since the outcome or actions of a smart contract will be as reliable as the data provided by the oracle. As of now, it seems there is no good solution for this trust problem, and blokchains have to rely on oracle services like Oraclize \cite{Oraclize}.
%\footnote{\url{http://www.oraclize.it}, accessed Nov 22, 2018.}.

\subsection{Contributions}

As mentioned above, in this paper we extend \cg{the} study of the distributed ledger reliable concurrent object DLO started in \cite{FGKN_NETYS18} to systems formed of several such objects. Hence, the first contribution is the definition of the Multiple DLO (MDLO) system, as the aggregation of several DLOs (in similar way as a Distributed Shared Memory is the aggregation of multiple registers \cite{DBLP:books/daglib/0032304}).
The second contribution is the definition of a simple basic problem
in MDLO systems: the \emph{2-AtomicAppends problem.} In this problem, the objective is that two records belonging to two different clients are appended to two different DLOs atomically. Hence, either both records are appended or none is. Of course, this problem can be generalized in a natural way to the $k$-Atomic Appends problem, involving $k$ clients with $k$ records and up to $k$ DLOs.

Another contribution, in our view, is the introduction of a crash-prone risk-averse rational client model, which we believe is natural and practical, \cg{especially
	in the context of blockchains.} In this model, clients act selfishly trying to maximize their utility, but minimizing the risk of reducing it. We consider that this behavior is not a failure, but the nature of the client, and any algorithm proposed under this model (e.g., to solve the 2-AtomicAppends problem) must guarantee that clients will follow it, because their utility will be maximized without any risk. For a complete specification of the clients' rationality their utility function has to be provided. Two utility models are proposed. In the \emph{collaborative utility model}, both clients want the records to be appended over any other alternative. In the \emph{competitive utility model} a client still wants both records appended, but she prefers that only the other client appends.
This client model is complemented with the possibility that clients can fail by crashing.

We explore hence the solvability of 2-AtomicAppends in MDLO systems in which the DLOs are reliable but may be asynchronous, and the clients are rational but may fail by crashing. %as described (see Table~\ref{tbl:summary}). 
The first results we present consider a system model in which clients do not crash, and show that Collaborative 2-AtomicAppends can be solved even under asynchrony, while Competitive 2-AtomicAppends cannot be solved.
Then, we further study Collaborative 2-AtomicAppends if clients can crash. In the case that at most one of the two clients can crash, we show that, if each client must append its own record (what we call \emph{no delegation}), Collaborative 2-AtomicAppends cannot be solved even under synchrony. This justifies exploring the possibility of \emph{delegation:} any client can append any record, if she knows it. We show that in this case Collaborative 2-AtomicAppends can be solved, even if the system is asynchronous (termination is only guarantee under synchrony, though). However, delegation is not enough if both clients can crash, even under synchrony.
\cg{(See Table~\ref{tbl:summary} for an overview.)}

The negative results (for Competitive 2-AtomicAppends even without crash failures and for Collaborative 2-AtomicAppends with up to 2 crashes) justifies exploring alternatives to appending directly direct or delegating among clients. Hence, we propose the use of an entity, external to the clients, that coordinates the appends of the two records. In fact, this entity is a special DLO with some level of intelligence, which we hence call {\em Smart DLO} (SDLO). The SDLO is a reliable entity to which clients can delegate (via appending in the SDLO) the responsibility of appending their records to their respective DLOs when convenient. The SDLO hence collects all the records from the clients and appends them. Since the SDLO is reliable, all the appends will complete. If some record is missing, the SDLO issues no append, to guarantee the properties of the 2-AtomicAppends problem. Thus, the SDLO can be used to solve Competitive and Collaborative $k$-AtomicAppends even when all clients can crash. 

We believe that SDLO opens the door to a new type of interconnection and interoperability among DLOs and blockchains. While the use of oracles to access external information in a smart contract (maybe from another blockchain) is widely known, we are not familiar with blockchain systems in which one blochchain (i.e., possibly a smart contract) issues transactions in another blockchain. We believe this is a concept worth to be explored further.

The rest of the paper is structured as follows. The next section describes the model used and defines the AtomicAppends problem. Section~\ref{sec:nocrash} explores the 2-AtomicAppends problem when clients cannot crash. Section~\ref{nosdlo} studies the 2-AtomicAppends problem when clients can crash but SDLOs are not used. Section~\ref{sdlo} introduces the SDLO and shows how it solves the AtomicAppends problem.
Finally, Section~\ref{conclusion} presents conclusions and future work.

%% file: model.tex
%\section{Discussion in Madrid October 2018}
\subsection{Objects and Histories}
An object type $T$ is defined over the domain of values that any object of type $T$ may take,
and the operations that any object of type $T$ supports. 
An object $\obj{}$ of type $T$ is a \emph{concurrent object} if it is a shared object accessed by 
multiple processes~\cite{Raynal13}.
A \emph{history} of operations on an object $\obj{}$, denoted by $\hist{\obj{}}$, is the sequence of operations invoked 
on $\obj{}$. Each operation $\op$ contains an \emph{invocation} and a matching \emph{response} \af{event}. 
Therefore, a \emph{history}
is a sequence of invocation and response \af{events}, starting with an invocation.  We say that an operation $\op$  
is \textit{complete} in a history $\hist{\obj{}}$, if  the history contains both the invocation and the \emph{matching} response \af{events}
of $\op$. History  $\hist{\obj{}}$ is \emph{complete} if it only contains complete operation.  History $\hist{\obj{}}$ is \emph{well-formed} if no two invocation events that do not have a matching response 
event in $\hist{\obj{}}$ belong to the same process $\pr$. That is, each process $\pr$ invokes one operation at a time.
An object history $\hist{\obj{}}$ is \emph{sequential}, if it contains a sequence of alternating invocation and 
matching response events, starting with an invocation and ending with a response. We say that an operation 
$\op_1$ \emph{happens before} an operation $\op_2$ in a history $\hist{\obj{}}$, denoted by $\op_1\bef\op_2$,
if the response event of $\op_1$ appears before the invocation event of $\op_2$ in $\hist{\obj{}}$.\medskip 

\noindent{\bf The Ledger Object (LO).}
A \emph{ledger} $\ledger$ (as defined in \cite{FGKN_NETYS18}) is a concurrent object that stores a totally ordered sequence $\ledger.S$ of \emph{records} 
and supports two operations (available to any process $\pr$):
(i) $\ledger.\act{get}_\pr()$, and (ii) $\ledger.\act{append}_\pr(\rdr)$.
A \emph{record} is a triple $\rdr=\tup{\tau, \pr, v}$, where $\pr$ is the identifier of the process that created record $\rdr$,
$\tau$ is a {\em unique} record identifier from a set ${\mathcal T}$, 
and $v$ is the data of the record drawn from an alphabet \af{$\Sigma$.} %$\valSet$.  
We will use $\rdr.\pr$ to denote the id of the process that created record $\rdr$; similarly we define $\rdr.\tau$ and $\rdr.v$.
A process $\pr$ invokes an $\ledger.\act{get}_\pr()$ operation
%\footnote{We define only one operation to access the value of the ledger for simplicity. In practice, other operations, like those to access individual records in the sequence, will also be available.} 
to obtain the sequence $\ledger.S$ of records stored in the ledger object $\ledger$, and $\pr$ invokes an $\ledger.\act{append}_\pr(\rdr)$ operation to extend $\ledger.S$ with a new record $r$.
Initially, the sequence $\ledger.S$ is empty. 

\begin{definition}[Sequential Specification of a LO \cite{FGKN_NETYS18}]
	\label{def:sspec}  
	The \emph{sequential specification} of a ledger $\ledger$ over the sequential history $\hist{\ledger}$ is defined as follows. The value of the sequence $\ledger.S$ of the ledger is initially the empty sequence.
	If at the invocation event of an operation $\op$ in $\hist{\ledger}$ the value of the sequence in ledger $\ledger$ is $\ledger.S=V$, then: 
	\begin{enumerate}
		\item  if $\op$ is an $\ledger.\act{get}_\pr()$ operation, then the response event of $\op$ returns $V$, \af{while the value of $\ledger.S$ does not change,} and
		\item if $\op$ is an $\ledger.\act{append}_\pr(\rdr)$ operation \af{(and $r \notin V$),}
		% that returns {\sc ack}, 
		then at the response event of $\op$ the value of the sequence in ledger $\ledger$ is $\ledger.S=V\extends r$ (where $\extends$ is the concatenation operator).
		% and
		%\item if $\op$ is an $\act{append}(r)_\pr$ operation that returns {\sc %nack}, 
		%then at the response event of $\op$ the value of the ledger is $\ledger$.
	\end{enumerate}
\end{definition}

In this paper we assume that ledgers are \emph{idempotent}, therefore a record $r$ appears only once in the ledger even when the same record $r$ 
is appended to the ledger by multiple $\act{append}$ operations (and hence the $r\notin V$ in the definition above). 

%\nn{[NN: If we assume idempotent records then the property 2 of the definition should change. No?]}\\
%\cg{CG: I would simply say that we consider implementations of DLOs (in the%
%	next section) where if a record is already in the ledger, then it is%
%	not appended again (but the append operation returns ack).}

\subsection{Distributed Ledger Objects (DLO) and Multiple DLOs (MDLO)}

\noindent{\bf Distributed Ledger Objects (DLO).} 
A {\em Distributed Ledger Object (DLO)} $\dledger$, is a concurrent LO that is \textit{implemented} by (and possibly replicated among) 
a set $\srvSet$ of (possibly distinct and geographically dispersed) computing devices, we refer as \emph{servers}. Like any LO,
$\dledger$ supports the operations $\act{get}()$ and $\act{append}()$. 
We refer to the processes that invoke the $\act{get}()$ and $\act{append}()$ operations on $\dledger$ as \emph{clients}.

Each server $s\in\srvSet$ may fail. Thus, the distribution and replication of $\dledger$ offers availability and survivability of the ledger 
in case a subset of servers fail. At the same time, the fact that multiple clients invoke $\act{append}()$ and $\act{get}()$ requests 
to different servers, raises the challenge of \emph{consistency}: 
what is the latest value of the ledger when multiple clients access the ledger concurrently? The work in \cite{FGKN_NETYS18}   
defined three consistency semantics to explain the behavior of $\act{append}()$ and $\act{get}()$ operations 
when those are invoked concurrently by multiple clients on a single DLO. In particular, they defined \emph{linearizable}~\cite{Lynch1996,HW90}, 
\emph{sequential}~\cite{LL79}, and \emph{eventual}~\cite{MSlevels} consistent DLOs. In this work we will focus on \emph{linerizable}
DLOs which according to \cite{FGKN_NETYS18} are defined as follows: 

\begin{definition} [Linearizable Distributed Ledger Object \cite{FGKN_NETYS18}]
	\label{def:atomic}
	A distributed ledger $\dledger$ is {\em linearizable} if, given any complete, well-formed history 
	$\hist{\dledger}$, there exists a \af{sequential} permutation $\sigma$ of the operations in $\hist{\dledger}$ such that: 
	\begin{enumerate}
		\item $\sigma$ follows the sequential specification of a ledger object (Definition \ref{def:sspec}), and 
		\item for every pair of operations $\pi_1, \pi_2$, if $\pi_1\bef \pi_2$ in $\hist{\dledger}$, then $\pi_1$ appears before $\pi_2$ in $\sigma$. 
	\end{enumerate}
\end{definition}

\noindent{\bf Multiple DLOs (MDLO).}
%We observe as well that currently in blockchains smart contracts (as far as we know) cannot have interactions outside the blockchain itself. We believe very interesting to extend somehow smarts contract so they can access and modify the environment somehow (they only exception we found is Cosmos and PolkaDot). In fact we will use this below to achieve interoperable DLOs.
%\af{[AF: Move to related work?]} \cg{[CG: Or in motivation?]}\nn{[NN: Yes has to go in intro somewhere.]}
%
%In particular, we consider 
A {\em Multi-Distributed Ledger Object} $\mdledger$, termed MDLO, consists of a collection $D$ of (heterogeneous) DLOs
and supports the following operations: (i) $\mdledger.\act{get}_\pr(\dledger)$, and (ii) $\mdledger.\act{append}_\pr(\dledger,\rdr)$.
The $\act{get}$ returns the sequence of records $\dledger.S$, where $\dledger\in D$. Similarly, the $\act{append}$ operation 
appends the record $\rdr$ to the end of the sequence $\dledger.S$, where $\dledger\in D$.
%Similarly, a MVDLO is a collection of VDLOs. Observe that if 
From the locality property of linearizability~\cite{HW90} it follows
that a MDLO is linearizable, if it is composed of linearizable DLOs. More formally:
%%(resp. MVDLO) 
%contains only DLOs 
%%(resp. VDLOs) 
%with linearizable consistency, then the whole MDLO (resp. MVDLO) is linearizable, as long as the access to the MDLO (resp. MVDLO) by the clients is well formed; this follows
%from the locality property of linearizability~\cite{RaynalBook18}. From this point onwards we only consider linearizable DLOs.
%%and VDLOs.\\ \cg{[CG: As we discussed, lets consider only DLOs for now (let VDLOs for later).]}\af{[I agree]}\nn{[NN: OK.]}

\begin{definition} [Linearizable Multi-Distributed Ledger Object]
	\label{def:atomic}
	A multi-distributed ledger $\mdledger$ is {\em linearizable} if
	\af{$\forall \dledger\in D$, $\dledger$ is 
		linearizable, where $D$ is the set of DLOs $\mdledger$ contains.}
	%, given any complete, well-formed history 
	%$\hist{\mdledger}$, if $D$ is the set of DLOs it contains, then 
\end{definition}

For the rest of this paper, \af{unless otherwise stated,} we will focus on MDLOs consisting of two DLOs, as the same techniques can be  
generalized in MDLOs with more than two DLOs.  In particular, we consider the records of two clients, 
$A$ and $B$, on two different DLOs. For convenience we use DLO$_X$ to denote the DLO 
appended by records from $X$, for $X\in\{A,B\}$. Similarly we denote as $r_X$ the record that $X\in\{A,B\}$ 
wants to append on $DLO_X$. Furthermore, we view the DLOs and MDLOs as black boxes that reliably implement
the specified service, without going into further implementation details. 

%We will consider both {\em synchronous} and {\em asynchronous} systems; in the former, 
\subsection{AtomicAppends: Problem Definition}
\label{subsec:prob-def}

Multi-DLOs allow clients to interact with different DLOs concurrently. This is safe when the records involved in concurrent operations 
are independent. However, it may raise semantic consistency issues when there exists a inter-dependent records,
e.g. a record $r_A$ must be inserted in DLO$_A$ when a record $r_B$ is inserted in DLO$_B$ and vice versa. 
More formally, we say that a record $r$ \emph{depends} on a record $r'$, if $r$ may be appended on its intended DLO, say $\dledger$,
only if $r'$ is appended on a DLO, say $\dledger'$. 
In this section we define a new problem, we term 
\emph{AtomicAppends}, that captures the properties we need to satisfy when multiple operations attempt to append 
dependent records on different DLOs. 
%
%We begin by defining the {\em $2$-AtomicAppends} problem.

\begin{definition}[$2$-AtomicAppends]
	\label{def:2AA}
	Given two clients, $A$ and $B$, with mutually dependent records $r_A$ and $r_B$. 
	%that have to be appended to $DLO_A$ and $DLO_B$, respectively, of a multi-DLO $\mdledger$. 
	%Client $A$ cannot
	%Append $r_B$ in \DLOA and $B$ cannot append $r_A$ in \DLOB\footnote{For example, because they mutually do not know these records.}. 
	We say that records $r_A$ and $r_B$ are {\em appended atomically} on DLO$_A$ and DLO$_B$ respectively, 
	when:% clients $A$ and $B$ execute protocol $P$ such that:
	\begin{itemize}
		\item Either both records are appended to their respective DLOs, or none of the records is appended (safety).
		\item If neither $A$ nor $B$ fail, %and follow protocol $P$, 
		then both records are appended eventually (liveness). 
	\end{itemize}
\end{definition}

An algorithm {\em solves} the $2$-AtomicAppends problem under a given system model, if it guarantees the safety and
liveness properties of Definition~\ref{def:2AA}.

%Observe that the safety condition must hold even if one or both $A$ and $B$ crash or do not follow protocol $P$ (i.e., even if they are malicious).
%{\bf CG: The definition of the problem is general wrt the failure model, i.e., covers any type of failure, but for our initial investigation of the problem, as we discussed, it seems natural  to consider rational clients that might fail by crashing.}\\

The $k$-\emph{AtomicAppends} problem, for $k\ge 2$, is a generalization of the $2$-AtomicAppends that can be defined in the natural way ($k$ clients, 
with $k$ records, to be appended to up to $k$ DLOs.) From this point onwards, we will focus on the $2$-AtomicAppends problem, and when clear from the context, 
we will refer to it simply as \emph{AtomicAppends}.

\subsection{Communication, Timing and Append Models}
The previous subsections are independent of \af{the communication medium, and the failure and timing model.} 
We now specify the communication and timing assumptions considered in the remainder of the paper.
We also consider different models on who can append a specific record.\medskip

\noindent {\bf Communication model:} We assume a {\em message-passing} system where messages are neither lost nor corrupted in transit. This applies
to both the communication among clients and between clients and DLOs \af{(i.e, the invocation and response messages of the operations)}.\medskip 

\noindent {\bf Timing models:} We consider {\em synchronous} and {\em asynchronous} systems with respect to both 
computation and communication. In the former, the evolution of the system is governed by a global clock and a local computation, a message delivery or a DLO operation is guaranteed to complete within a predefined time-frame. For simplicity, we set this time-frame to correspond to one unit of time.
In the latter, no timing assumptions are made \af{beyond that they will complete in a finite time}.\medskip

\noindent {\bf Append models:} We consider three different append models. In the first, and most restrictive one, which we refer to as {\em Client appends with no delegation}, or {\bf\em NoDelegation} for short, the only way a client can append its record, is by issuing append operations directly to the corresponding DLOs, i.e., no other entity, including the other client, can do so. The second one, referred to as \emph{Client appends with delegation},  
or {\bf\em WithDelegation} for short, is a relaxation of the first model, in which one client can append the record of the other client (if it knows it). Finally, in the third model, a record can be appended by an external (w.r.t. the clients) entity, provided it knows the record. 

\subsection{Client Model and Utility-based Problem Definitions}

\subsubsection{Client Setting}
We assume that clients are {\bf\em rational},
i.e., they act selfishly, in a game-theoretic sense, in order to increase their utility~\cite{osborne2004introduction}. %This implies that a client does not 
%follow a given protocol that reduces its utility. 
Furthermore, clients are {\bf\em risk-averse}, 
i.e.,
when uncertain, they prefer to lower the uncertainty,
%\nn{[NN: what do we mean here by ``lower the uncertainty''? Can we measure the 
%	uncertainty? is there a risk utility function that measures the risk and/or uncertainty? Maybe the utility function of each 
%	client should be a function of its gain and its risk.]}, 
even if this might lower their \af{potential} utility~\cite{osborne2004introduction}\cg{; we consider a client to be uncertain when 
	her actions may lead to multiple possible outcomes.}   
To this respect, a rational, risk-averse client runs its own utility-driven protocol that defines its strategy
	towards a given protocol (game), in such a way that it would not decrease its utility  or increase its uncertainty. 
	
	Regarding failures, 
	%As long as a rational client does not deviate from its own strategy, then this client is considered correct. %This means that a client might not execute Note that the above behavior is not considered a failure. 
%
the only type of failure we consider in this work, is {\bf\em crash failure}, in which a 
client might cease operating without any a priori warning.

Under this client model, {\bf\em an algorithm ${\mathcal A}$ solves the AtomicAppends problem}, if it provides
enough incentive to the clients to follow this algorithm (which guarantees the  safety and
liveness properties of Definition~\ref{def:2AA}, \af{possibly in the presence of crashes}), 
without any client deviating from its utility-driven protocol.
If no such algorithm can be designed, then {\bfseries\em the AtomicAppends problem cannot be solved}.

\subsubsection{Utility Models}
Looking at the definition of the AtomicAppends problem, one might wonder what is the incentive of the clients to achieve this
both-or-none principle on the appends. Let $U_X$ denote the utility function (or incentive) for each client $X$.
%A rational client $X$ would like to increase her utility $U_X$, while a selfish rational client $X'$ will target to maximize her utility $U_X'$. 
\af{A selfish rational client $X$ will try to maximize her utility $U_X$.}
%\nn{[NN: Rationality needs to be defined.]} \af{[I think at this point both rationality and utility can be left vaguely used, but we do need a formal explicit statement of what model of clients we are using: selfish rational clients prone to crash. Then, the selfishness implies maximizing her utility. This goes into the system model.]} 
%increase its utility \nn{[NN: utility needs to be defined]}; hence, depending on the situation
%it might be to its best interest to achieve this, but in fact it might prefer this not to be achieved. 
Depending on the 
\af{possible combinations of values the clients' utility functions can take,}
%strategy that each client uses to set its utility function, 
we can identify a number 
of different scenarios, we refer as \textit{utility models}. 
% We will refer to different situations as {\em utility models}. 
Let us now motivate and specify two such utility models.\medskip 

\noindent{\bf Collaborative utility model.} Consider two clients $A$ and $B$ that have agreed to acquire a property (e.g., a piece of land) in common, and each has to provide half of the cost. If one of them, say $A$, pays while $B$ backs off from the deal, then $A$ incurs 
\af{in}
%additional 
expenses while not getting the property. On the other hand, $B$ loses no money in this case, but her reputation may suffer. If both of them back off, they do not have any cost, while if both proceed with the payments then they get the property, which they prefer.

If $U_X()$ denotes the utility of agent $X \in \{A,B\}$, then we have the following relations in the scenario described: 
$$\small{U_X(\text{both agents pay}) > U_X(\text{no agent pays}) 
> U_X(\text{only agent }  \bar{X} \text{ pays}) > U_X(\text{only agent }  X \text{ pays}).}$$

In relation to the AtomicAppends problem, \af{record} $r_A$ contains the transaction by which  client $A$ pays her share of the deal, and the append of $r_A$ in \DLOA 
%represents client $A$ paying her share of the deal. 
carries out this payment. Similarly for client $B$. So, here we see that under the above utility model, 
both clients have incentive for both appends to take place. Observe that this situation is similar to the {\em Coordinated Attack} problem~\cite{CoordinatedAttack}, 
in which two armies need to agree on attacking a common enemy. If both attack, then they win; if only one of them attacks, then that army is destroyed, while the other is disgraced;
if none of them attack, then the status quo is preserved. 

These utility examples fall in the general utility model depicted
in the first row of Table~\ref{tbl:utilmodels}, which we
call {\em collaborative}.
%So, in general, any situation in which having both clients appending maximizes their utilities, 
%falls under this utility model, 
%which we call .    
We will be referring to the AtomicAppends problem under this utility model as the 
{\bf\em Collaborative AtomicAppends} problem.\medskip
 
\begin{table}[t]
	\begin{small}
		\centering
		\begin{tabular}{|c|c|}
			\hline
			{\bf Utility model} & {\bf Utility of client $\mathbf{X}$}\\
			\hline\hline \vspace{-.9em}\\            
			Collaborative & $U_X(\text{both append}) > U_X(\text{none appends})>$\\ 
			& $ U_X(\text{only }  \bar{X} \text{ appends}) > U_X(\text{only }  X \text{ appends})$\\ 
			\hline \vspace{-.9em}\\ 
			Competitive & $U_X(\text{only }  \bar{X} \text{ appends}) > U_X(\text{both append}) >$\\
			& $U_X(\text{none appends}) 
			>  U_X(\text{only }  X \text{ appends})$\\
			\hline
		\end{tabular}\smallskip
		\caption{\small The utility of client $X\in\{A,B\}$ in the two utility models considered.}
		\label{tbl:utilmodels}
	\end{small}
\end{table}

\noindent{\bf Competitive utility model.} We now consider a different utility model. 
Consider two clients $A$ and $B$ that have agreed to exchange their goods. E.g, $A$ gives his car to $B$, and $B$ gives a specific amount as payment
to $A$. If one of them, say $A$, gives the car to $B$, but $B$ does not pay, then $A$ loses the car while not getting any money. 
On the other hand, $B$ gets the car for free! If both of them back off from the deal, then they do not have any cost. Both proceeding with the exchange is not necessarily their highest preference  (unlike in
the previous collaborative model).

So, if $U_X()$ denotes the utility of agent $X \in \{A,B\}$, then we have the following relations in the scenario described: 
$$\small{U_X(\text{only } \bar{X} \text{ proceeds}) > U_X(\text{both agents proceed}) > U_X(\text{no agent proceeds}) > U_X(\text{only } X \text{ proc.}).}$$

In relation to the AtomicAppends problem, \af{record} $r_A$ contains the transaction transferring the deed of $A$'s car to $B$, and the append of $r_A$ in \DLOA 
carries out this transfer. Similarly, $r_B$ contains the transaction by which client $B$ transfers a specific monetary amount to $A$ (pays for the car), 
and the append of $r_B$ in \DLOB carries out this monetary transfer. Observe that this scenario is similar to the {\em Atomic Swaps} problem~\cite{DBLP:conf/podc/Herlihy18}.

These utility examples fall in the general utility model depicted
in the second row of Table~\ref{tbl:utilmodels}, which we
call {\em competitive}. We will be referring to the AtomicAppends problem under this utility model as the 
{\bf\em Competitive AtomicAppends} problem.

%So in general, under this utility model it falls any situation in which the two clients have higher preference for their record not to be appended while the other's is appended.
%Lesser beneficial is when both records are appended, and the least preference is when no record is being appended. We call this utility model {\em competitive}. As %before, for ease of presentation, we will be referring to the AtomicAppends problem under the competitive utility model as the 
{%\bf\em Competitive AtomicAppends} problem. Table~\ref{tbl:utilmodels}
%ummarizes the utility relation of the two introduced versions of the AtomicAppend problem.

No matter of the utility,  failure or timing model assumed, our objective  is to provide a solution to the AtomicAppends problem. 
%(consider it as a fairness property). 
Our investigation will focus on identifying the modeling conditions under which this is possible or not, and what is the
impact of the model on the solvability of the problem.

\section{AtomicAppends in the Absence of Client Crashes}   
\label{sec:nocrash}
We begin our investigation in a setting with no client crashes, so to study the impact of the utility model
on the solvability of the problem.

\sloppy{It is not difficult to observe that in the absence of crash failures, even under asynchrony and NoDelegation, there is a straightforward algorithmic solution to the {\em Collaborative AtomicAppends} problem: the algorithm simply has client $A$ (resp. client $B$) issuing operation $append(DLO_A,r_A)$ (resp.  $append(DLO_B,r_B)$). Based on Table~\ref{tbl:utilmodels}, the clients' utilities are maximized when both append their corresponding records. Since there are no failures and the DLOs are reliable, these
	operation are guaranteed to complete, nullifying the clients' \cg{uncertainty}.
	%\nn{[NN: Here we talk about risk. I believe we need to be consistent and talk either about risk or uncertainty (depending on how we will define it in the model)]}. 
	Hence, the clients will follow the algorithm, without deviating from
	their utility-driven protocol. This yields the following result:
	
	\begin{theorem}
		\label{thm:collaborativeNoFailures}
		Collaborative 2-AtomicAppends can be solved in the absence of failures, even under asynchrony and NoDelegation.
	\end{theorem}

	However, this is {\em not} the case for the \emph{Competitive AtomicAppends} problem. The problem
	cannot be solved, even in the absence of failures, in synchrony, and WithDelegation:

	\begin{theorem}
		\label{thm:competitiveNoFailures}
		Competitive 2-AtomicAppends cannot be solved in the absence of failures, even in synchrony and WithDelegation.
	\end{theorem}
	
	\begin{proof}
	Let us firstly show that client $A$ will never send its record $r_A$ to the other client $B$. The reason is that this would carry a large risk of $B$ appending $r_A$ itself (and $A$ is risk-averse). Observe that, independently on whether $B$ already appended $r_B$ or not, this would reduce $A$'s utility (see Table~\ref{tbl:utilmodels}). 	
	Then, we secondly claim that client $A$ will not directly append its own record $r_A$ either. The reason is that, again, independently on whether $B$ already appended $r_B$ or not, this would reduce $A$'s utility (see Table~\ref{tbl:utilmodels}).
	Hence, client $A$ will not have its record $r_A$ appended to $DLO_A$ ever. However, this violates the liveness property of 
	Definition~\ref{def:2AA}, since by assumption neither $A$ nor $B$ fail by crashing.\end{proof}	
%		Case analysis: First argue that they will not exchange their records, and then that they will not append their
%		own record. This will lead to a violation of liveness.

	Note that the above result does not contradict the known solutions for atomic swaps (e.g.,~\cite{DBLP:conf/podc/Herlihy18}), as the primitives used
	are stronger than the ones offered by DLO (e.g., some form of validation is needed for hashlocks). 
	As we show in Section~\ref{sdlo}, the problem can be solved in the model we consider, if a reliable external entity is used between the
	clients and the MDLO. 
	In view of Theorems~\ref{thm:collaborativeNoFailures} \af{and \ref{thm:competitiveNoFailures},} in the next section we focus on the
	study of {\em Collaborative AtomicAppends} in the presence of crash failures.
	
	%So, from this point on, we assume that at least one client crashes.}

%% file: collaborative_noSDLO.tex
%\subsection{Collaborative AtomicAppends Problem in MDLO}

In this section we focus on the Collaborative AtomicAppend problem \af{assuming
	that at least one client may crash,
under the NoDelegation and WithDelegation client append models.}
%and consider a model in which no external intermediate entity is used by the clients to solve the problem; we refer to this model as {\em Direct Appends}. 
Observe from Table~\ref{tbl:utilmodels} that both clients have incentive to get both records appended, versus the case of
no record appended, with respect to utilities.
However, as we will see, in some cases, crashes introduce uncertainty that renders the problem unsolvable.

\subsection{Client Appends with No Delegation} %Consider now a synchronous system in which clients can crash. 
We prove  that \emph{Collaborative AtomicAppends} cannot be guaranteed by any algorithm ${\mathcal A}$, even in a {\em synchronous system},
when at least one client crashes and the clients cannot delegate the append of their records.

\begin{theorem}
	\label{thm:impossibleDirect}
When at least one client crashes, Collaborative 2-AtomicAppends cannot be solved in \af{the NoDelegation append model,} even in a synchronous system.
\end{theorem}
\begin{proof}
%Consider that there are two clients, $A$ and $B$, with records $r_A$ and $r_B$ respectively, that have to be appended to \DLOA and \DLOB, respectively.
%
Consider an algorithm ${\mathcal A}$ that clients can execute without deviating from their utility-driven protocol. Assume algorithm ${\mathcal A}$
solves the Collaborative 2-AtomicAppends problem in the model described.   
Let $E$ be an execution of algorithm ${\mathcal A}$ in which no client crashes. %Based on Table~\ref{tbl:utilmodels} and b
By liveness, both clients $A$ and $B$ must issue append operations. Consider the first client, say $A$ without loss of generality, that issues the append operation. Let us assume that $A$ issues $append(\mathrm{DLO}_A,r_A)$ at time $t$. Hence, $B$ issues $append(DLO_B,r_B)$ at time no earlier than $t$, and $A$ cannot verify that the record $r_B$ is in the corresponding $DLO_B$ until time $t' > t$.

Now consider execution $E'$ of  algorithm ${\mathcal A}$ that is identical to $E$,  {\em up to time} $t$. Now at time $t$ client $B$ crashes, and hence it never issues $append(DLO_B,r_B)$. Since $A$ cannot differentiate until time $t$ this execution from $E$, it issues $append(DLO_A,r_A)$ at time $t$, appending $r_A$ to $DLO_A$. Even if after time $t$, $A$ detects the crash of client $B$, by the specification of NoDelegation, it cannot append record $r_B$ in \DLOB. This, together with the fact that $B$ has crashed, yields that record $r_B$ is never appended to $DLO_B$, violating safety. 
%If client $A$ does not append record $r_A$ in execution  $E'$ at time $t$, then since it cannot distinguish it from
%execution $E$, it will not append in $E$ either, violating liveness or deviating from its utility-driven protocol (which cannot do). %\vspace{1em}
Hence, we reach a contradiction, and algorithm~${\mathcal A}$ does not solve the Collaborative 2-AtomicAppends problem. 
\end{proof}

\subsection{Client Appends With Delegation}
Let us now consider the more relaxed client append model of WithDelegation. 
%a relaxation of the model in which one client can append the record of the other client (if it knows it). We call this model \emph{Direct Appends with delegation}.
It is not difficult to see that in this model, the  impossibility proof of Theorem~\ref{thm:impossibleDirect} breaks. 
In fact, it is easy to design an algorithm that solves the collaborative AtomicAppends problem in a synchronous system, if at most one client crashes.
In a nutshell, first both clients exchange their records. When a client has both records, it appends them (one after the other) to the corresponding DLO; otherwise it does not append any record. We refer to this algorithm as Algorithm $\mathcal{A}_{DSync}$ and
its pseudocode is given as Code~\ref{alg:DSync}. We show:

\remove{%Placed in algorthmic format.
\begin{itemize}
	\item In the first step $A$ sends the record $r_A$ to $B$ and simultaneously $B$ sends record $r_B$ to $A$.
	\item After this step, if client $X \in \{A,B\}$ has both records, it appends $r_X$ in \DLOX and then it appends $r_{\bar{X}}$ in \DLOBX. Otherwise, it does not append any record.
\end{itemize} 
}

\begin{theorem}
	\label{thm:possibleDirectSync}
	In \af{the WithDelegation append model,} Algorithm $\mathcal{A}_{DSync}$ solves the Collaborative 2-AtomicAppends problem in a synchronous system, if at most one client crashes.
\end{theorem}

\begin{proof}
If no client crashes, then the proof of the claim is straightforward. Hence, let us consider the case that one client crashes, say $A$. There are three cases: 
\begin{enumerate}[label=(\alph*)]
	\item Client $A$ crashes
	before sending its \af{record.} In this case, client $B$ will not append any record and the problem is solved (none case).
	\item  Client $A$ crashes after sending its record,
	but before it does any append. In this case client $B$ will receive \af{$A$'s record and append both records} (both case).
	\item Client $A$ crashes after it performs
	one or two of the appends. Client $B$ will perform both appends, and since DLOs guarantee that a record is appended only once (they are idempotent), the problem is solved (both case).
\end{enumerate}
The above cases and Table~\ref{tbl:utilmodels} suggest that the clients have no risk in running Algorithm $\mathcal{A}_{DSync}$ with respect to their utility-driven
protocol. Hence, the claim follows.%\vspace{1em}
\end{proof}

\begin{algorithm}[t]
	\caption{\small $\mathcal{A}_{DSync}$: \af{AtomicAppends WithDelegation, Synchrony,} at most one crash; code for Client $X\in \{A,B\}$.}
	\label{alg:DSync}
	\begin{algorithmic}[1]
		\State \textbf{send}  $r_X$ to client $\bar{X}$
		\State \textbf{If} $r_{\bar{X}}$ is received from client $\bar{X}$ \textbf{then}
		\Indent 
		\State \textbf{append}$(DLO_{X},r_X)$
		\State \textbf{append}$(DLO_{\bar{X}},r_{\bar{X}})$
		\EndIndent
		\State \textbf{Else} (client $\bar{X}$ has crashed) 
		\Indent 
		\State \textbf{no append}
		\EndIndent   
		%\Statex
	\end{algorithmic}
\end{algorithm}

\begin{algorithm}[t]
	\caption{\small ${\mathcal A}_{DAsync}$: \af{AtomicAppends WithDelegation, Asynchrony,} at most one crash; code for Client $X\in \{A,B\}$.}
	\label{alg:DAsync}
	\begin{algorithmic}[1]
		\State \textbf{send}  $r_X$ to client $\bar{X}$
		\State \textbf{wait} until 
		$r_{\bar{X}}$ is received from client $\bar{X}$ \textbf{or}
		a \textbf{get}$(DLO_{\bar{X}})$ returns ${S}$ such that 
		$r_{\bar{X}}\in S$ (i.e., $r_{\bar{X}}\in DLO_{\bar{X}}.S$)
		\Indent 
		\State \textbf{append}$(DLO_{X},r_X)$
		\State \textbf{if} $r_{\bar{X}}\not\in DLO_{\bar{X}}.S$ \textbf{then} \textbf{append}$(DLO_{\bar{X}},r_{\bar{X}})$
		\EndIndent    
		%\Statex
	\end{algorithmic}
\end{algorithm}

Observe that, as written, Algorithm $\mathcal{A}_{DSync}$ is not well-suited for asynchrony, since a client cannot
distinguish the case on whether the other client has crashed or its message is taking too long to arrive.
For this purpose we modify  Algorithm $\mathcal{A}_{DSync}$ and obtain Algorithm $\mathcal{A}_{DAsync}$ 
which is better-suited for an asynchronous system. In a nutshell, first a client sends its record to 
the other client, and then it waits until either it receives the record of the other
client, or that record is already appended (by the other client) in the corresponding DLO (this would
not have happened in a synchronous system).
 When one of the two hold, the client proceeds to append the two records (one after the other).
 The pseudocode of algorithm $\mathcal{A}_{DAsync}$ is given as Code~\ref{alg:DAsync}.%\medskip 
 We show:

\remove{%Placed in algorithmic formal
\begin{itemize}
	\item []  Each client $X$ does:
	\begin{enumerate}
	\item Send own record $r_X$ to client $\bar{X}$.
	\item Wait until record $r_{\bar{X}}$ is received or it has been
	appended to DLO$_{\bar{X}}$ (by $\bar{X}$) [issues periodic $get()$ operations to 
	DLO$_{\bar{X}}$]
	\item Append record $r_X$ to DLO$_{X}$
	\item Append record $r_{\bar{X}}$, if not yet appended, to DLO$_{\bar{X}}$.  
	\end{enumerate}
\end{itemize} 
}

%We will be referring to the above as Algorithm $DAsync$. Then:

\begin{theorem}
	\label{thm:possibleDirectAsync}
	In \af{the WithDelegation append model,} Algorithm $\mathcal{A}_{DAsync}$ solves the Collaborative 2-AtomicAppends problem in an asynchronous system, if at most one client crashes.
\end{theorem}

\begin{proof}
	As before, we will prove this by case analysis. 
	If no client crashes, then the proof follows easily, given the fact
	that a DLOs guarantees that a record is appended only once. 
	 Hence, let us consider the case that one client crashes, say $A$. There are three cases: 
	\begin{enumerate}[label=(\alph*)]
		\item Client $A$ crashes
		before sending its \af{record.} In this case, client $B$ will not
		proceed to append any record (none case). Observe that client $B$
		might not terminate, but the problem (safety) is not violated.
		\item  Client $A$ crashes after sending its record,
		but before it does any append. In this case client $B$ will receive \af{$A$'s record and append both records} (both case).
		\item Client $A$ crashes after it performs
		one or two of the appends (it means it has sent its record to client $B$). Client $B$ will perform both appends, and since DLOs guarantee that a record is appended only once, the problem is solved (both case).
	\end{enumerate}
	The above cases and Table~\ref{tbl:utilmodels} suggest that the clients have no risk in running Algorithm $\mathcal{A}_{DAsync}$ with respect to their utility-driven
	protocol. Hence, the claim follows.%\vspace{1em}
\end{proof}

As already discussed in case (a) of the above proof, it is possible for the client that has not crashed to wait forever, as it cannot distinguish the case \af{when} the message is taking too long to arrive and the append operation is taking too long \af{to complete, from the case when} the other client has crashed. Hence, algorithm $\mathcal{A}_{DAsync}$, under certain conditions, is {\em non-terminating}.

Furthermore, it is not difficult to see that if both clients fail, neither
algorithm $\mathcal{A}_{DAsync}$ nor algorithm$\mathcal{A}_{DSync}$ can solve the Collaborative AtomicAppends problem. For example, in the proof of Theorem~\ref{thm:possibleDirectSync}, in case (b), client $B$ could crash 
right after appending its own record (i.e., $r_B$ is appended, but $r_A$ is not). This violates safety. In fact, we now show that if both clients can crash, the problem
is not solvable, even under synchrony.

\begin{theorem}
	\label{thm:impos2crashes} 
	When both clients can crash, the Collaborative 2-AtomicAppends problem cannot be solved WithDelegation, even in a synchronous system.
\end{theorem} 

\begin{proof}
%	Consider an execution $E$ without failures. Then both records are appended. Assume the first record is appended at time $t$. Consider another run $E'$ equal to $E$ up to time $t$, after which both clients are crashed. Then only the first record is appended, violating safety.%\vspace{1em}
Consider an algorithm ${\mathcal A}$ that clients can execute without deviating from their utility-driven protocol. Assume algorithm ${\mathcal A}$ solves the Collaborative 2-AtomicAppends problem in the model described.   
Let $E$ be an execution of algorithm ${\mathcal A}$ in which no client crashes. %Based on Table~\ref{tbl:utilmodels} and b
By liveness, both records $r_A$ and $r_B$ must be eventually appended. Consider the first record appended, say $r_A$ w.l.o.g., and the client that issued the append operation, say $A$ w.l.o.g.. Let us assume that $A$ issues $append(\mathrm{DLO}_A,r_A)$ at time $t$. Hence, $append(DLO_B,r_B)$ is issued at time no earlier than $t$, and $A$ cannot verify that the record $r_B$ is in the corresponding $DLO_B$ until time $t' > t$.
	
Now consider execution $E'$ of  algorithm ${\mathcal A}$ that is identical to $E$,  {\em up to time} $t$. Now at time $t$ client $B$ crashes, and hence it never issues $append(DLO_B,r_B)$. Since $A$ cannot differentiate until time $t$ this execution from $E$, it issues $append(DLO_A,r_A)$ at time $t$, appending $r_A$ to $DLO_A$. 
Then, at time $t+1$ (immediately after $append(DLO_A,r_A)$ completes) $A$ also crashes, and hence never issues $append(DLO_B,r_B)$.
Since $append(DLO_B,r_B)$ is never issued, record $r_B$ is never appended to $DLO_B$, violating safety.
Hence, we reach a contradiction, and algorithm~${\mathcal A}$ does not solve the Collaborative 2-AtomicAppends problem. 
\end{proof}

%end of remove

%% file: collaborative_SDLO.tex
%This result motivates the need of a more complex MDLO to be able to handle this and other operations. 
%Observe, that the coordinated appends problem, although seems similar to atomic swaps (cf \cite{DBLP:conf/podc/Herlihy18}), it is in fact
%different. [Explain difference.]   
%For atomic swaps, which is related to atomic appends, one could employ techniques like those described by Herlihy \cite{DBLP:conf/podc/Herlihy18} which require synchrony and timelocks.

%In view of the impossibility result of the previous section, 
Theorems~\ref{thm:competitiveNoFailures} and \ref{thm:impos2crashes} suggest the need to use some  external intermediary entity,
in order to solve {\em Competitive AtomicAppends}, even in the absence of crashes, and {\em Collaborative AtomicAppends}, in the case both
clients crash, respectively. This is the subject of this section.
\subsection{Smart DLO (SDLO)}
We enhance the MDLO with a special DLO, called {\em Smart DLO} (SDLO),
which is used by the clients to delegate the append of their records to the original MDLO. 
% \af{[AF: remove?: (The get operations are done directly on the MDLO.)]} \cg{[CG: Yes.]}
This SDLO is an extension of a DLO that supports a special ``atomic appends'' record of the form {\bf [client id, $\{$list of involved clients in the atomic append$\}$, record of client]}. When two clients wish to perform an atomic append involving their records
and their corresponding DLOs, then they \emph{both} need to append such an atomic appends
record in the SDLO; this is like requesting the atomic append service from the
SDLO. Once \emph{both} records are appended in the SDLO, then the SDLO appends each record to the corresponding DLO. A pseudocode of this mechanism, together with the client requests, called algorithm ${\mathcal A}_{SDLO}$ is given as Code~\ref{alg:DSDLO}.
 
\begin{algorithm}[h]
	\caption{\small ${\mathcal A}_{SDLO}$: SDLO mechanism and requests from clients $A$ and $B$; SDLO code only for atomic appends}
	\label{alg:DSDLO}
	\begin{algorithmic}[1]
		\State\textbf{Client $A$:}
		\Indent
		\State   \textbf{append}$(SDLO,[A,\{A,B\},r_A])$
		\State   \textbf{upon} receipt {\rm AppendAck} from SDLO \textbf{return} 
		\EndIndent
		\Statex
		\State\textbf{Client $B$:}
		\Indent
		\State   \textbf{append}$(SDLO,[B,\{A,B\},r_B])$
		\State   \textbf{upon} receipt {\rm AppendAck} from SDLO \textbf{return} 
		\EndIndent
		\Statex
		\State \textbf{SDLO:}
		\Indent
		\State \textbf{Init:} ${S}\gets \emptyset$
	    \Function{SDLO.\act{append}}{$[X,\{X,\bar{X}\},r_X]$}
	    \State $S \gets S~\extends~[X,\{X,\bar{X}\},r_X]$
	    \If {$[\bar{X},\{X,\bar{X}\},r_{\bar{X}}]\in {S}$}
	    \State \textbf{append}$(DLO_X,r_X)$
	    \State \textbf{append}$(DLO_{\bar{X}},r_{\bar{X}})$
	    \EndIf
	    \State \textbf{return} AppendAck 
	    \EndFunction
		\EndIndent
		%\Statex
	\end{algorithmic}
\end{algorithm}

So essentially the SDLO ``collects'' the append requests involved in the AtomicAppends instance and 
ultimately executes them, by performing individual appends to the corresponding DLOs.
Observe that the SDLO does not access the state of \DLOA and \DLOB, but it needs to be able to perform append operations to both of them. In other words, delegation is passed to the SDLO. Also observe that the SDLO returns ack to a client's request, once their atomic appends request is appended in the SDLO,  and not when
the actual atomic append takes place (see related discussion below).

\remove{%placed in algorithmic forma
the mechanism in SDLO to implement atomic appends for two clients, $A$ and $B$, with records $r_A$ and $r_B$ respectively, that have to be appended to $DLO_A$ and $DLO_B$, respectively, would be as follows:
\begin{itemize}
\item 
Client $A$ invokes a function of a smart contract $AtomicAppends(a_A, h_B)$, where $a_A=append(DLO_A,r_A)$ and $h_B=hash(append(DLO_B,r_B))$.
\item 
Client $B$ invokes a function of a smart contract $AtomicAppends(a_B, h_A)$, where $a_B=append(DLO_B,r_B)$ and $h_A=hash(append(DLO_A,r_A))$.
\item
The function $AtomicAppends()$ detects when matching invocations have been appended to the SDLO, and it triggers the operations $append(DLO_A,r_A)$ and $append(DLO_B,r_B)$ to $DLO_A$ and $DLO_B$ respectively. 
%The records must be appended exactly once. How to guarantee this is an implementation issue that depends on the specific assumptions and model. We can for instance %asume that no DLO allows repeated records, and when multiple appends for the same record are received only one is applied. Then, whoever evaluates the %$AtomicAppends()$ and triggers the appends can do so safely.
%
%[To do: Specify which servers/miners actually does that. In principle, only one such server must do it, but there must be fault tolerance. This could be considered an implementation issue under specific assumptions under a system model.]
\end{itemize}
\af{[I would simplify the above description, removing hash functions. I think the record structure from the slides is simpler and hence better: $(A, \{A,B\}, r_A)$ ]}\cg{[Agreed.]}
}%end old description

% Hence, validity of the SDLO can be checked locally.

\subsection{Solving AtomicAppends with SDLO}
It is not difficult to observe that algorithm ${\mathcal A}_{SDLO}$ can 
solve the AtomicAppends problem in both utility models, even \emph{in asynchrony}, and even if \emph{both clients crash}. 
Note that $SDLO$, being a distributed ledger by itself, is reliable
%does not crash in its entirely, 
despite the fact that
some servers implementing it may fail (more below). We show: 

\begin{theorem}
	\label{thm:AA-SDLO}
	Algorithm ${\mathcal A}_{SDLO}$ solves both the Collaborative and Competitive 2-AtomicAppends problems in an asynchronous setting, even if both clients may crash.
\end{theorem}

\begin{proof}
	We consider three cases:
	\begin{enumerate}
		\item If no client crashes, then algorithm ${\mathcal A}_{SDLO}$ trivially solves the problem: Both clients invoke the
		atomic appends request to the SDLO, these operations complete, and the SDLO eventually triggers the two corresponding appends of records
		$r_A$ and $r_B$ to $DLO_A$ and $DLO_B$, respectively (both case).
		\item At most one client crashes, say client $A$. Here we have two cases:
		\begin{enumerate}
			\label{item:sdlo}
			\item Record $[A,\{A,B\},r_A]$ is never appended to the SDLO. Since the SDLO will never contain both matching records, it will never append any of the records $r_A$ and $r_B$ (none case).
			\item Record $[A,\{A,B\},r_A]$ is appended to the SDLO. Since record $[B,\{A,B\},r_B]$ will eventually be appended by $B$ in the SDLO,
			it will proceed with the corresponding appends of records $r_A$ and $r_B$ (both case). 
		\end{enumerate}
        \item Both clients crash. If one of the two clients, say $A$, crashes before appending $[A,\{A,B\},r_A]$ to the SDLO,
        then none of the appends of records
        $r_A$ and $r_B$ will take place in the corresponding DLOs (none case). However, if both clients crash after 
        they have appended the matching atomic appends records, then both records
        $r_A$ and $r_B$ will be appended by the SDLO (both case).  
	\end{enumerate}
   Observe that the above hold for both utility models. In Competitive AtomicAppends, if a
   client does not invoke its atomic append request to the SDLO, it knows that the SDLO will not proceed to append
   the other client's record. This leaves the clients with their second best utility (see Table~\ref{tbl:utilmodels}),
   and hence, both have incentive to invoke the atomic append requests to the SDLO. The reliability of the 
   SDLO nullifies the uncertainty of the clients, and hence they will follow algorithm ${\mathcal A}_{SDLO}$. %\vspace{1em}
\end{proof}

%Given that all the append operations go through the SDLO and all append operations in a DLO complete successfully, it is not difficult to see that the above mechanism %provides Atomic Appends in the presence of crash failures. [Check if it works for Byzantine as well.]

Observe that algorithm ${\mathcal A}_{SDLO}$ can easily be extended to solve the $k$-$AtomicAppend$ problem,
for any $k\ge 2$ (all clients submit their atomic append request to the SDLO, and then the SDLO performs the corresponding appends),
leading to the following corollary:

\begin{corollary}
	\label{cor:kAA-SDLO}
	Both the Collaborative and Competitive $k$-AtomicAppends problems can be solved with the use of SDLO in the asynchronous setting, even if all $k$ clients may crash.
\end{corollary}

\noindent {\bf\em Remark:} As we discussed in the case \ref{item:sdlo} of the proof of Theorem~\ref{thm:AA-SDLO}, if client $A$ crashes and record $[A,\{A,B\},r_A]$ is never appended to the SDLO, none of the records $r_A$ and $r_B$ will be appended. Now, observe that client $B$ can proceed to perform other operations once it has appended $[B,\{A,B\},r_B]$ (despite the fact that $r_B$ has not been appended to \DLOB,
as it is up to the SDLO to do so). Since clients do not need to wait forever
for any operation, % to take place (including of course the SDLO), 
algorithm ${\mathcal A}_{SDLO}$ is terminating with respect to the clients. Moreover, the SDLO also terminates the processing of all the operations, as long as the appends in other DLOs terminate.\medskip 

\noindent{\bf Implementation issues.} In the above mechanism and theorem, we treat the SDLO as one entity. Since, however,
the SDLO is a distributed ledger implemented by collaborating servers, there are some low-level implementation
details that need to be discussed. If we assume that the servers implementing the SDLO are prone to only
crash faults and that the SDLO is implemented using an Atomic Broadcast service, as described in~\cite{FGKN_NETYS18},
then algorithm ${\mathcal A}_{SDLO}$ can be implemented as follows: Clients $A$
and $B$ submit the atomic append requests to all servers implementing the SDLO. 
Once a server appends an atomic append request record to its local copy of the ledger,
it checks if the matching record is already in the ledger.
%receives the two matching append requests, it atomically broadcasts this information to all other servers (the underline atomic broadcast service guarantees that this information is delivered only once by every non-crash server, even if it is broadcasted by many servers, c.f.~\cite{FGKN_NETYS18}), 
If this is the case, it issues the two corresponding append operations for records $r_A$ and $r_B$.
If up to $f$ servers may crash, then it suffices that $f+1$ servers, in total, perform these append operations.
Given that each record is appended to a DLO at most once (the append operations are idempotent; if a record is already appended, it will not
be appended again), it follows that both records are appended in the corresponding DLOs.

% END REMOVE

%% file: conclusion.tex
%Concusions section
We have introduced the AtomicAppends problem, where given two (or more \af{in general}) clients,
each needs to append a record to a corresponding DLO, and do so atomically 
with respect to each other: either both records are appended or none.
We have considered crash-prone, rational and risk-averse clients based on two different utility
models, {\it Collaborative} and {\it Competitive}, and studied the solvability of the
problem under synchrony/asynchrony, different client append models
and failure \cg{scenarios}. %\nn{{[NN: Actually we only considered crash failures.]}}. 
Table~\ref{tbl:summary} gives an overview
of our results (for two clients): if the problem can be solved, then we list the
\textcolor{blue}{algorithm} we developed, otherwise we use the symbol ``\textcolor{red}{\xmark}''.
%Symbol ``{\bf \textcolor{green}{?}}" denotes that we are not aware of a solution. 

\begin{table}[h]
	\begin{small}
		\centering
\begin{tabular}{cc|c|c|c|c|c|c|}
	\cline{3-8}
	\multicolumn{1}{l}{} & \multicolumn{1}{l|}{} & \multicolumn{3}{c|}{\textbf{Synchrony}} & \multicolumn{3}{c|}{\textbf{Asynchrony}} \\ \cline{3-8} 
	\multicolumn{1}{l}{\multirow{-2}{*}{\textbf{}}} & \multicolumn{1}{l|}{} & {ND} & {WD} & {SDLO} & {ND} & {WD} & {SDLO} \\ \hline
	\multicolumn{1}{|c|}{} & \textit{no crashes} & {\color{blue} simple} & {\color{blue} } & {\color{blue} } & {\color{blue} simple} & {\color{blue} } & {\color{blue} } \\ \cline{2-3} \cline{6-6}
	\multicolumn{1}{|c|}{} & \textit{up to one} & {\color{red} } & \multirow{-2}{*}{{\color{blue} ${\mathcal A}_{DSync}$}} & {\color{blue} } & {\color{red} } & \multirow{-2}{*}{{\color{blue} ${\mathcal A}_{DAsync}^{(\star)}$}} & {\color{blue} } \\ \cline{2-2} \cline{4-4} \cline{7-7}
	\multicolumn{1}{|c|}{\multirow{-3}{*}{\textbf{Collaborative}}} & \textit{both} & \multirow{-2}{*}{{\color{red} \textbf{\xmark}}} & {\color{red} \textbf{\xmark}} & {\color{blue} } & \multirow{-2}{*}{{\color{red} \textbf{\xmark}}} & {\color{red} \textbf{\xmark}} & {\color{blue} } \\ \cline{1-4} \cline{6-7}
	\multicolumn{1}{|c|}{} & \textit{no crashes} & \multicolumn{2}{c|}{{\color{red} }} & {\color{blue} } & \multicolumn{2}{c|}{{\color{red} }} & {\color{blue} } \\ \cline{2-2}
	\multicolumn{1}{|c|}{} & \textit{up to one} & \multicolumn{2}{c|}{{\color{red} }} & {\color{blue} } & \multicolumn{2}{c|}{{\color{red} }} & {\color{blue} } \\ \cline{2-2}
	\multicolumn{1}{|c|}{\multirow{-3}{*}{\textbf{Competitive}}} & \textit{both} & \multicolumn{2}{c|}{\multirow{-3}{*}{{\color{red} \textbf{\xmark}}}} & \multirow{-6}{*}{{\color{blue} ${\mathcal A}_{SDLO}$}} & \multicolumn{2}{c|}{\multirow{-3}{*}{{\color{red} \textbf{\xmark}}}} & \multirow{-6}{*}{{\color{blue} ${\mathcal A}_{SDLO}$}} \\ \hline
	\multicolumn{8}{|l|}{$^{(\star)}$ might not terminate} \\ \hline%\smallskip
\end{tabular}
		\caption{\small Overview of the results. ND stands for NoDelegation and WD for WithDelegation.}
		\label{tbl:summary}
	\end{small}
\end{table}

Our results demonstrate a clear separation on the solvability of the problem
based on the utility model assumed when appends are done directly by the clients.
When appends are done using a special type of a DLO, which we call {\it Smart} DLO (SDLO), 
then the problem is solved in both utility models, even in asynchrony and even if both
clients may crash. 

Our investigation of AtomicAppends did not look into the semantics of the records being
appended. 
%An interesting extension of this work is to consider the AtomicAppends problem under malicious
%(Byzantine) clients. 
Consider, for example, the following scenario. Say that clients $A$ and $B$
initiate an atomic append request with records $r_A$ and $r_B$, respectively. While the atomic append request is 
being processed, say by the SDLO, client $B$ appends a record $r^\prime$ directly to $DLO_B$. 
It could be the case that the content of record $r^\prime$ is such, that it would affect record $r_B$. 
For example, say that the atomic
append involves the exchange of a deed of a car with bitcoins; record $r_A$ contains the transfer
of the deed and $r_B$ the transfer of bitcoins. If $r^\prime$ involves the withdrawal of bitcoins
from the wallet of client $B$, and this is appended first, then it could be the case that
the wallet no longer contains sufficient bitcoins to carry out the atomic appends request.
Even if we enforce the clients to perform all appends trough the SDLO (which practically speaking
is not desirable), still we need to {\em validate} records. Therefore, to tackle such cases, we will need to 
consider {\it validated} DLOs (VDLOs)~\cite{FGKN_NETYS18}. This is a challenging problem, especially in
asynchronous settings.

%\paragraph{Extensions to be done:}
%\begin{itemize}
%\item 
%Allow the clients to append the $DLO_A$ and $DLO_B$ while the Atomic Appends is happening.
%\item Construct in the SDLO complex transactions that use the state of $DLO_A$ and $DLO_B$ to decide the actions to be taken.
%\end{itemize}

%Copying from a message:
%%DLOs $A$ has money transactions, DLO $B$ has cars.
%Then we need to have an operation that forces that two records are appended, %one in each DLO, one moving the money and other transferring ownership of the %car. Either both or none of these records are appended, which means that we %have an atomic transaction (in the database sense) formed by the two transfers.\\

%[[CG: The above provides a type of dependence.]]